\documentclass[12pt]{article}

\usepackage{amssymb}
\usepackage{amsthm}
\usepackage{amsmath}

\newtheorem{theorem}{Theorem}
\newtheorem{lemma}[theorem]{Lemma}
\newtheorem{proposition}[theorem]{Proposition}

\theoremstyle{definition}
\newtheorem{definition}[theorem]{Definition}

\theoremstyle{remark}
\newtheorem{remark}[theorem]{Remark}

\numberwithin{equation}{section}\numberwithin{theorem}{section}

\def \<{\langle}
\def \>{\rangle}

\newcommand{\1}{{(1)}}

\newcommand{\C}{\mathbb{C}}

\newcommand{\N}{\mathbb{N}}
\newcommand{\Z}{\mathbb{Z}}
\def\noi{\noindent}

\newtheorem*{corollary*}{Corollary}
\newtheorem*{remark*}{Remark}
\newtheorem*{remarks*}{Remarks}

\begin{document}

\title{Quasifinite representations of the Lie superalgebra of quantum pseudo 
differential operators}

\author{Carina Boyallian\thanks{The  first author was supported in
part by
Conicet, ANPCyT, Agencia Cba Ciencia, and Secyt-UNC (Argentina).}
\  and Vanesa Meinardi}

\maketitle




\begin{abstract} In this paper we extend general results obtained in 
\cite{1} for quasifinite highest weight representations of $\Z$-graded 
Lie algebras to $\frac{1}{2}\Z$-graded Lie superalgebras, and we apply these 
  to classify the irreducible quasifinite highest weight modules of the Lie 
superalgebra of  quantum pseudo-differential operators.
\end{abstract}

\section{Introduction}

The study of Lie superalgebras and its representations plays an
important roll in Conformal Field Theory and Supersymmetries in
physics.

The main difficulty to develop a suitable representation theory
for certain $\Z$-graded Lie superalgebras lies in the fact that
the graded subspaces of some highest weight modules over these Lie
superalgebras are  infinite dimensional, in spite of having a
natural principal gradation and a triangular decomposition.
However, most of the physical theories usually require that these
subspaces have finite dimension, to which we refer as {\it
quasifinitness}.

In \cite{2}, Kac and Radul developed a powerful machinery and
begun the systematic study of  quasifinite representations of the
Lie algebra of differential operators on the circle and the Lie
algebra of quantum pseudo differential operators.

Following this  work, there have been later developments and many 
extensions as in  ([3],[4],[1],[5], etc). Moreover, the results in 
[2], were extended to the study for the quasifinite representations of the 
Lie superalgebra of differential operators on the supercircle in [6] and 
its  subalgebras in [7].

In [1], they developed a general theory that characterize the
quasifinite highest weight representation of any $\Z$-graded Lie
algebra, under some mild conditions.

In the first part of this article we extend this  results for Lie
superalgebras. Then using this, we classify the quasifinite
highest weight representation of the Lie superalgebra of quantum
pseudo differential operators. Observe that the extension of these
results are useful to simplify some computations made in [7] and
[8].

In order to give a realization of these representations in terms of tensor
products of quasifinite representations of the Lie superalgebra of
infinite matrices with a finite number of non-zero diagonals with
coefficients in the truncated polynomials, we need  to
characterize them, using the extension of the results in [1].

\section { Quasifinite representations of graded Lie superalgebras }

Recall that a \textit{superalgebra}  is a $\mathbb{Z}_2$-graded
algebra. A   \textit{Lie superalgebra } is a superalgebra
$\frak{g}=\frak g_{\bar{0}}\bigoplus \frak g_{\bar{1}}$ ,
$(\overline{0},\overline{1}\in \Z_2)$, with multiplication given
by a \textit{ super bracket} $[\,,\,]$ satisfying:
$$[a,b]=-(-1)^{\bar{\alpha}\bar{\beta}}[b,a]$$
$$[a,[b,c]]=[[a,b],c]+(-1)^{\bar{\alpha}\bar{\beta}}[b,[a,c]],$$
for all $a\in {\frak{g}_{\bar{\alpha}}}$,
$b\in \frak{g}_{\bar{\beta}}$ with $\bar{\alpha}$ and $\bar{\beta}\in\Z_2$.

Let $\frak{g}$ be $\frac{1}{2}\Z$-graded Lie superalgebra over $\mathbb{C}$, 
namely
$${\frak g=\bigoplus_{j\in \frac {1}{2}\Z}{ \frak g_j}}\quad
\hbox{ and }\quad[\frak g_i,\frak g_j]\subseteq\frak
g_{i+j} \hbox{ with } i,\,j \in \frac{1}{2}\mathbb{Z}.$$ The
$\frac{1}{2}\mathbb{Z}$-gradation of a Lie superalgebra is
\textit{consistent} with the
$\textit{Z}_2$-gradation if
$$
\frak g_{\bar{0}}=\bigoplus_{i\in \Z}{\frak g_{i}} \hbox{ and}
\quad
\frak g_{\bar{1}}=\bigoplus_{i\in \Z}\frak {g}_{i+\frac{1}{2}}.
$$
For a $\frac{1}{2}\mathbb{Z}$-graded Lie superalgebra $\frak{g}$,
set
$$
\frak{g}_+=\bigoplus_{i\in \frac{1}{2}{\Z}_{>0}}\frak{g}_i , \quad
\frak{g}_-=\bigoplus_{i\in\frac {1}{2}{\Z}_{>0}}{\frak{g}_{-i}}.
$$
In this section  $ \frak g$ will denote a consistent
$\frac{1}{2}\Z$-graded Lie superalgebra.

\begin{definition} A {subalgebra} $\frak  p$ of $\frak g$ is called 
\textit{parabolic }if it contains  ${\frak{g}_0}\oplus{\frak{g}_+}$ as a 
proper subalgebra, that is,
$$
\frak{p}=\bigoplus_{j \in \frac{1}{2}\Z} \frak{p}_j,\, \hbox{
where } \, \frak{p}_j= \frak{g}_j \, \quad\hbox{for}
\quad j\geq
{0}
\, \hbox{ and } \,\frak p_{-j}\neq {0} \, \hbox{ for some }\,
{j}>{0}.
$$

\end{definition}

\noindent We assume the following properties on  $\frak g:$

\,

$(SP_1)$ $\frak g_0$ is commutative,

\,

$(SP_2)$ If ${a} \in  \frak g_{-k}$ $({k}>{0})$, \,and $[a,\frak
g_{\frac{1}{2}}]=0,$ then $a=0.$

\begin{remark}\label{lemm p1} As an immediate consequence of the definition 
of parabolic subalgebra and condition ($SP_2$), if $ \frak p$  is any parabolic 
subalgebra of $\frak g$ with $\frak p_{-k}\neq 0$ $(k>0),$ then 
$\frak{p}_{-k+{\frac{1}{2}}}\neq 0$.
\end{remark}

Given $a\in \frak{g}_{-{\frac{1}{2}}}$, ${{a}\neq{0}}$, we define
\,$\frak{p}^{a}=\bigoplus_{j\in
{\frac{1}{2}}\Z}{\frak{p}_{j}^{a}},$ where
\begin{align}
&\quad\quad\frak{p}_j^{a}=\frak{g}_j  \quad\hbox{ for all } \quad{j\geq 
0},\nonumber\\
&\quad\quad
\frak{p}_{-\frac{1}{2}}^{a}=\sum[\cdots[[a,\frak{g}_0],\frak{g}_0],\cdots]\nonumber\\
&\hbox{
and}\,\quad{{\frak{p}_{-k-\frac{1}{2}}^{a}=[\frak{p}_{-\frac{1}{2}}^{a},
\frak{p}_{-k}^{a}]}}.
\end{align}

We have the following Lemma, whose proof we shall omit since it is
identical to the proof of Lemma 2.2 in [1] with the obvious
modifications.

\begin{lemma}\label{lemm p2} Let $a \in \frak g_{-\frac {1}{2}}, a\neq 0.$ 
Then:

(a) $\frak{p}^{a}$ is the minimal parabolic superalgebra containing $a$.

(b) $\frak g_0^{a}:={[\frak p^{a},\frak p^{a}]}\bigcap{\frak
g_0}=[a,\frak g_{\frac{1}{2}}].$
\end{lemma}

\begin{remark} The examples of parabolic subalgebras considered in [7] y 
[8] motivate the following definition.
\end{remark}

\begin{definition} (a) A parabolic subalgebra  $\frak{p}$
is called  \textit{non-degenerate} if $\frak{p}_{-j}$ has finite codimension 
in $\frak{g}_{-j}$  for all \, $j>0.$

(b) A non-zero element  $a\in \frak{g}_{-\frac {1}{2}}$ is called
\textit{non-degenerate} if $\frak{p}^{a}$ is non-degenerate.
\end{definition}

In order to study quasifinite representations of graded Lie
superalgebras we recall some definitions and notions.

A $\frak{g}${-module} $ {V}$ is called $\frac
{1}{2}\mathbb{Z}$\textit{-graded} if
$${V}=\bigoplus_{j\in \frac {1}{2} \mathbb{Z}}V_j \quad\hbox{ and 
}\quad{\frak g_iV_j\subseteq{V_{i+j}}} \quad (i,j \in \frac {1}{2}
\Z),
$$
and $V$ is called \textit{quasifinite} if \hbox{ dim
}${V_j<{\infty}}$ \, {for all} $j$.

  Given $\lambda\in{{\frak{g}}_0}^{\ast}$, a  \textit{highest weight module} 
is a $\frac{1}{2}\Z-$graded $\frak g$-module
$V(\frak g,\lambda)=\bigoplus_{j\in \frac{1}{2}\Z}V_{j}$ defined
by the following properties:
\begin{align}
V_0&=\C v_{\lambda} \quad\hbox{where} \, v_{\lambda} \,\hbox{is a non zero 
vector,} \nonumber\\
hv_\lambda&=\lambda(h)v_\lambda \quad \hbox{for}\quad \, h\in
\frak g_0, \nonumber\\
\,\frak g_+v_\lambda&=0,\nonumber\\
\,{\cal U}(\frak g_{-})v_{\lambda}&=V(\frak g,\lambda).
\end{align}
Here and further ${\cal U}(\frak s)$ stands for the universal
enveloping superalgebra of the Lie superalgebra $\frak s$.

A non-zero vector  $v\in V(\frak g,\lambda)$ is called
\textit{singular} if $\frak g_+v=0.$

The \textit{ Verma module} is constructed as follows
$$M(\frak g,\lambda)={\cal{U}}(\frak 
g)\bigotimes_{{\cal{U}}(\frak{g}_{0}\bigoplus{\frak{g}_+})}{\C_\lambda},$$
where $\C_\lambda:= \mathbb{C}c_\lambda $, is the $1$-dimensional
$\frak{g}_0 \oplus\frak{g}_+$-module given by $h c_\lambda=
\lambda(h)c_\lambda$ if $h\in \frak g_0$ and $\frak g_{+}
c_\lambda= 0,$
  and the action of $\frak g$ on $M(\frak g,\lambda)$ is induced by the left 
multiplication
in ${\cal U}(\frak g)$.

Any highest module  $V(\frak g,\lambda)$ is a quotient of $M(\frak
g,\lambda).$ The "smallest" among the  $V(\frak g,\lambda)$ is the
unique irreducible module  $L(\frak g,\lambda)$ (which is the
quotient of $M(\frak g,\lambda)$ by its maximal graded submodule).

\noi For simplicity we denote  ${V(\frak{g},\lambda)}=
{V(\lambda)},$ \,$ {M(\frak{g},\lambda)}={M(\lambda)}$ and
${L(\frak{g},\lambda)}= {L(\lambda)}$.

Now , let $\frak{p}=\bigoplus_{j\in \frac {1}{2}\mathbb{Z}}{\frak
p_j},$ be a parabolic subalgebra of $\frak g$ and let $\lambda\in
\frak{g}_0^{\ast}$  be such that
  $\lambda|_{{{\frak{g}_0}\bigcap{[\frak p,\frak p]}}=0}$.
Then the $\frak g_0 \oplus \frak g_+$-module
$\mathbb{C}_\lambda=\mathbb{C}c_\lambda$ extends to $\frak p$ by
letting $\frak p_j\cdot c_\lambda= 0$ for $j< 0,$ and we may
construct the highest weight module
$$M(\frak
g,\lambda,\frak{p})={\cal{U}(\frak{g})}\bigotimes_{\cal{U}(\frak{p})}{\C_\lambda},$$
which is called the \textit{generalized Verma module}.

\pagebreak

We  also require the following condition on $\frak g$:

$(SP_3)$ If $\frak p$ is a non-degenerate parabolic subalgebra of
$\frak g,$ then there exists a non-degenerate element  $a$ such
that
$\frak p^a\subseteq \frak p$.
\begin{remark} The examples considered in [7] and [8] satisfy the 
properties
$(SP_1),\,(SP_2)$ y $\,(SP_3)$.\end{remark}

The main result of this section is the following Theorem, whose
proof we shall omit since it is completely analogous to the one in
[1].

\begin{theorem}\label{teo:quasi} Let $\frak g$ be a Lie  superalgebra that 
satisfies $(SP_1),\,(SP_2)$ and $(SP_3)$.
The following conditions on $\lambda\in g_0^{*}$  are equivalent:

\,

(1)\,$M(\lambda)$   contains a singular vector $av_\lambda\in
M(\lambda)_{-\frac {1}{2}}$ where  $a$ is non-degenerate;

(2)\,There exist a non-degenerate element $a\in \frak{g}_{-\frac
{1}{2}}$  such that $\lambda([\frak{g}_{\frac {1}{2}},a])=0.$

(3)\,$L(\lambda)$ is  quasifinite.

(4)  There exist a  non-degenerate element $a\in \frak{g}_{-\frac
{1}{2}}$  such that $L(\lambda)$  is the irreducible quotient of
the generalized Verma module $M(\frak{g},\lambda,\frak{p}^{a}).$
\end{theorem}

\section{The Lie superalgebra of quantum pseudo differential operators}

Let $q\in \C^{\times}$ and $ |q|\neq 1 $. Now, $T_q$ denote the following operator on
$\C[z,z^{-1}]$:
$$
T_q(f(z))=f(qz)
$$
Let $\frak{S}_{q}^{as}$
denote the associative algebra of all  operators on
  $\C[z,z^{-1}]$ of the form
$$E=\sum_{k\in\Z}{e_{k}(z)T_{q}^k} \quad\hbox { where }  e_{k}(z)\in 
\C[z,z^{-1}]
\hbox{ and the sum is finite}.$$ We write such an operator as a
linear combination of operators of the form $z^{k}f(T_q)$, where
$f$ is a
Laurent polynomial in $T_q$ and $k\in \Z$.
The product in $\frak{S}_{q}^{as}$ is given by
$$
(z^{m}f(T_q))(z^{k}g(T_q))= z^{m+k}f(q^{k}T_q)g(T_q).
$$

Denote by $M(1|1)$ the set of $2\times2$ supermatrix
$$
\begin{pmatrix}
  f_{11} & f_{12} \\
  f_{21} & f_{22}
\end{pmatrix},$$
where $f_{ij}\in\C$, viewed as the associative superalgebra of
linear transformations of the complex $(1|1)$-dimensional
superspace $\C^{1|1}$.

We  denote $M_{ij}$  the $2\times 2$ matrix
with $1$ in the $ij$-place  and $0$ everywhere else. Declaring $M_{11}$, 
$M_{22}$  even and $M_{12}$, $M_{21}$ odd elements, we endow $M(1|1)$ with a 
$\Z_2$-gradation.

We denote by $S\frak S_q^{as}$  the associative superalgebra  of
$2\times 2$ supermatrices with entries in $\frak S_q^{as}$ namely $$
S\frak S_q^{as}=\frak S_q^{as}  \otimes M(1|1),
$$ and
the product is given by the usual matrix multiplication.  Let
$S\frak S_q$ denote the  corresponding Lie superalgebra   where
the  Lie superbracket is explicitly given by:
\begin{align}
[z^nf(T_q)M_{ij},\,z^mg(T_q)M_{rs}\}&=(z^nf(T_q)M_{ij})(z^mg(T_q)M_{rs})\nonumber\\
&\qquad-(-1)^{|M_{ij}||M_{rs}|}(z^mg(T_q)M_{rs})(z^nf(T_q)M_{ij})\nonumber\\
\,&=z^{n+m}\Big(f(q^mT_q)g(T_q)\delta_{jr}M_{is}\nonumber\\
&\qquad-(-1)^{|M_{ij}||M_{rs}|}g(q^nT_q)f(T_q)\delta_{si}M_{rj}\Big),
\end{align}
where $|M|$ denotes the parity of $M$. Now, introduce the linear
map \linebreak $Str_0:$  $S\frak S_q$ $\rightarrow$ $\C$ as

$$
Str_0 \begin{pmatrix}
  f_{11}(T_q) & f_{12}(T_q) \\
  f_{21}(T_q) & f_{22}(T_q)
\end{pmatrix}=\big({f_{11}(T_q)}\big)_0-{\big(f_{22}(T_q)\big)}_0,
$$
where $\big(f(T_q)\big)_0=f_0$ if $f(T_q)=\sum_k f_kT_q^k$,
$(f_k\in\mathbb{C})$. We should notice that  $Str_0$ has the
following property:
$$
Str_0(f(T_q)M_{ij}\,g(T_q)M_{kl})=(-1)^{|M_{ij}||M_{kl}|}Str_0(g(T_q)M_{kl}\,f(T_q)M_{ij}).
$$
Thus,  define a one-dimensional central extension $\widehat{S\frak
S_q}$ of $ S\frak S_q$ with the following super bracket:
$$[z^rF(T_q),z^sG(T_q)]=[(z^rF(T_q)),(z^sG(T_q))\}
+\psi_{\sigma,Str_0}(z^rF(T_q),z^sG(T_q))C,
$$
where $C$ is the {\textit{ central charge}, and  the  super
$2-cocycle$ $\psi_{\sigma,Str_0}$ is given by
\begin{align}
\psi_{\sigma,Str_0}(z^rf(T_q)M_{ij},z^{s}&g(T_q)M_{kl})=\nonumber\\
&Str_0\Big((1+\sigma+\cdots+\sigma^{r-1})\big(\sigma^{-r}(f(T_q)M_{ij})g(T_q)M_{kl}\big)\Big)
\nonumber\\
\,&=-(-1)^i\sum_{m=0}^{r-1}\Big(f(q^{-r+m}T_q)\,g(q^mT_q)\Big)_0\delta_{kj}\delta_{il}.
\end{align}
if  $r=-s>0$ and $0$ otherwise. Here $\sigma$ is the  automorphism of  $\frak S_q^{as}$ given by
$\sigma(f(T_q)M_{ij})=f(qT_q)M_{ij}$ (cf. with (1.3.1) in [2]).

The \textit{  principal} $\frac{1}{2}\Z$-{\textit{gradation}} in
$\widehat{S\frak S_q}$ is given by $\widehat{S\frak
S_q}=\bigoplus_{\alpha\in\Z/2}\widehat{(S\frak S_q)}_{\alpha}$, ($n\in\mathbb{Z}$),
\begin{equation}\label{grad1}
{ \widehat{(S\frak
S_q)}_{\alpha=n}=\{z^n\big(f_{11}(T_q)M_{11}+f_{22}(T_q)M_{22}\big)+\delta_{n,0}C
\,:\,{f_{ii}\in\C[w,w^{-1}]} \hbox{ with } \,i=1,2\}},
\end{equation}
and
\begin{align}\label{grad2}
\widehat{(S\frak S_q)}_{\alpha=n+1/2}=
\{z^nf_{12}(T_q)M_{12}+z^{n+1}f_{21}(T_q)M_{21}\,:\,
&{f_{ij}\in\C[w,w^{-1}]},\nonumber\\
&\quad{i,j\in \{1,2\}}\hbox{ with }{i\neq j} \}.
\end{align}

\section{Quasifinite Representations of $\widehat{S{\frak S}_q}$}

Let $V(\lambda)$ be a highest weight module over $\widehat{S\frak
S_q}$ with  highest weight $\lambda$. The highest weight vector
$v_\lambda \in V(\lambda)$  is characterized via the principal
gradation as $\widehat{(S\frak S_q)}_{\alpha}v_\lambda=0$ for
$\alpha\geq 1/2$ and $\widehat{(S\frak S_q)}_0 v_\lambda \in \C
v_\lambda$. Explicitly, these conditions are written as:
\begin{align*}
z^nf(T_q)M_{ij}v_\lambda=0 \, &\quad\hbox{with }  {n\geq 1},
\,{f(w)
\in \C[w,w^{-1}]}, \,{i=j} {\hbox{ or }} \,{i=1, \,j=2}; \nonumber\\
z^{n+1}f_{21}(T_q)M_{21}v_\lambda=0 \, &\quad{\hbox{with }} \,
{n\geq 0}, \,{f_{21}(w)
\in \C[w,w^{-1}]}; \nonumber\\
f_{12}(T_q)M_{12}v_\lambda=0 \, &\quad{\hbox{ with }} \,
\,{f_{12}(w)
\in \C[w,w^{-1}]};\nonumber\\
(T_q^s M_{ii})v_\lambda=\lambda(T_q^sM_{ii})v_\lambda \,
&\quad{\hbox{ with }} \, {s \in \Z}  ,\,{i=1,2}.
\end{align*}


Consider $\frak p=\bigoplus_{\alpha\in \Z/2} \frak p_\alpha${ a
parabolic subalgebra of $\widehat{S\frak S_q}$. Thus $\frak
p_\alpha=(\widehat {S\frak S_q})_\alpha$ for all $\alpha\geq 0$
and ${\frak p_{\alpha}\neq 0}$ for some $\alpha<0.$
Observe that for each $j \in \mathbb{N}$ we have
$$ \frak p_{-j}=\{z^{-j}(f_{11}(T_q)M_{11}+f_{22}(T_q)M_{22})+\delta_{j,0}C: 
f_{ii}(w)\in I_{-j}^{ii}\,{\hbox{with}} \,i=1,2\}$$
$$ \frak p_{-j+1/2}=\{z^{-j}f_{12}(T_q)M_{12}+z^{-j+1}f_{21}(T_q)M_{21} 
:f_{rs}(w)\in I_{-j}^{rs}
{\hbox{ with }} \,r,s\in\{1,2\}, \, r\neq s\},$$
where $I_{-j}^{rs} \hbox{ with }  r,s \in \{1,2\} $  are subspaces
of $\C[w,w^{-1}].$ Since $[\widehat{(S\frak S_q)}_0,\frak
p_{-\alpha}]\subseteq \frak p_{-\alpha}  \hbox{ with}\,
\alpha\in\frac{1}{2}\mathbb{N}$, it is easy to check that $I_{-k}^{rs}$
satisfies
$$
A_{k}^{rs}I_{-k}^{rs}\subseteq I_{-k}^{rs} \quad\hbox{with}\quad
r,s\in \{1,2\},$$ where $\,A_{k}^{rs}=\{f(q^{-k}w)-f(w) \,:\,
f(w)\in \C[w,w^{-1}]\}$ if $ r,s \in \{1,2\}$ and $ r=s$ or $ r=1$
and $ s=2$, and $A_{k}^{21}=\{f(q^{-k+1}w)-f(w) \,:\, f(w)\in
\C[w,w^{-1}] \}$.

\begin{lemma}\label{prop ideal}
(a)\,$I_{-k}^{rs}$ is an ideal $\hbox{ for all } {k\in\N}$ and
${r,s \in \{1,2\}}.$

(b) \,If $I_{-k}^{rs}\neq 0$ then it has finite codimension in
$\C[w,w^{-1}].$
\begin{proof}
Since $|q|\neq 1$, observe that $A_k^{rs}=\C[w,w^{-1}] \hbox{ for
all } {k\geq \frac{1}{2}},\,{r,s \in \{1,2\}}$. Then $I_{-k}^{rs}$
is an ideal. Let $b_{-k}^{rs}$ be the monic polynomials that
generate the corresponding ideals $I_{-k}^{rs}$, therefore
$\hbox{dim}
(\C[w,w^{-1}]/\Big((b_{-k}^{rs})=I_{-k}^{rs})\Big)<\infty.$
\end{proof}
\end{lemma}

\begin{proposition}\label{prop p1}
(a)\, any non-zero element of $\widehat{(S\frak S_q)}_{-1/2}$ is
non-degenerate.

(b) \, Any parabolic subalgebra of $\widehat{S\frak S_q}$ is
non-degenerate.

(c)\,Let $d=z^{-1}b_{12}(T_q)M_{12}+b_{21}(T_q)M_{21}\in
(\widehat{S \frak S_q)}_{-1/2}.$ Then:
\begin{align}
\,\widehat{(S \frak S_q)}_{0}^d:&=[\widehat{(S \frak
S_q)}_{1/2},d]\nonumber\\
\,&=\{f(T_q)b_{21}(T_q)I+g(q^{-1}T_q)b_{12}(T_q)M_{22}+b_{12}(qT_q)
g(T_q)M_{11}-\nonumber\\
\,&   \qquad(g(q^{-1}T_q)\,b_{12}(T_q))_0 \,C\,:\, f(w),\,g(w)\in
\C[w,w^{-1}]\}.
\end{align}
\end{proposition}
\begin{proof}
Let $0\neq d \in \widehat{(S \frak S_q)}_{-1/2}.$ Then, by Lemma
\ref {prop ideal} (b), part (a) follows. Let $\frak p$ be any
parabolic subalgebra of $\widehat{S \frak S_q},$ using Remark
\ref{lemm p1} we get $\frak p_{-1/2}\neq 0.$ Then, using (a) and
$\frak p^d\subseteq \frak p$ for any non-zero $ d\in \frak
p_{-1/2}$, we obtain (b). Let $
d=z^{-1}b_{12}(T_q)M_{12}+b_{21}(T_q)M_{21}\in \widehat{(S \frak
S_q)}_{-1/2},$ and $ a=f(T_q)M_{12}+z\,g(T_q)M_{21}\in \widehat{(S
\frak S_q)}_{1/2},$ with  $ b_{ij}(w)$, $f(w)$ and $ g(w)\in
\C[w,w^{-1}]$, where $i,j\in \{1,2\}$ with $i\neq j$. Then
\begin{align}
\,[a,d]&=[f(T_q)M_{12},b_{21}(T_q)M_{21}]+[zg(T_q)M_{21},z^{-1}b_{12}(T_q)M_{12}]\nonumber\\
\,&=f(T_q)b_{21}(T_q)M_{11}+b_{21}(T_q)f(T_q)M_{22}\nonumber\\
\,&\quad 
+g(q^{-1}T_q)b_{12}(T_q)M_{22}+b_{12}(qT_q)g(T_q)M_{11})-\big(g(q^{-1}T_q)\,b_{12}(T_q)\big)_0\,C\nonumber\\
\,&=f(T_q)b_{21}(T_q)I+g(q^{-1}T_q)b_{12}(T_q)M_{22}+b_{12}(qT_q)
g(T_q)M_{11}\nonumber\\
&\hskip 8cm-\big(g(q^{-1}T_q)\,b_{12}(T_q)\big)_0 \,C,
\end{align}
finally, part c) follows by  Lemma \ref{lemm p2} (b).
\end{proof}

A functional $\lambda \in (\widehat{S \frak S_q})_0^\star$ \,is
described by its  \textit{labels}
$$\Delta_{l,i}=-\lambda(T_q^lM_{ii}),  $$ $\hbox{ where }  i=1,2,
\hbox{ and }   l\in \Z,$ and the central charge $\, c=\lambda(C).
$ We shall consider the generating series
$$\Delta_{\lambda,i}(x)=\sum_{l \in \Z}x^{-l}\Delta_{l,i} \quad i= 1,2.$$

Recall that a \textit{quasipolynomial } is a linear combination of
functions of the form $p(x)e^{\alpha x} \hbox{ where } \, p(x)\in
\C[x]$ and $\alpha\in \C$. A formal power series is a
quasipolynomial if and only if it satisfies a non-trivial linear
differential equation with constant coefficients. We also have the
following well known result.

\begin{theorem}\label{teo:quasipoly} Given a   quasipolynomial $q(x)$ and a 
polynomial $B(x)=\prod_i(x-A_i)$,
let $b(x)=\prod_i(x-a_i) \hbox{ where } a_i=e^{A_i}$. Then
$\,b(x)(\sum_nq(n)x^{-n})=0$ if and only if $B(d/dx)q(x)=0$.
\end{theorem}

Now we state the main result of this article.

\begin{theorem}\label{teo:QFM} An irreducible highest weight module 
$L(\widehat{S\frak
S_q},\lambda)$ is quasifinite if and only if one of the following
equivalent conditions hold:

(i)\,There exist two monic non-zero polynomials
$b_{12}(x),\,b_{21}(x)$ such that
\begin{align}
\,&b_{12}(x)(\Delta_{\lambda,1}(q^{-1}x)+\Delta_{\lambda,2}(x)-c)=0,\nonumber\\
\,&b_{21}(x)(\Delta_{\lambda,1}(x)+\Delta_{\lambda,2}(x))=0.
\end{align}

(ii)\,There exist quasipolynomials  $P_{12}(x)$ and $P_{21}(x)$ such
that $P_{21}(0)=P_{12}(0)+c$ and   ($n\in \Z$, $n\neq 0$):
\begin{align}
P_{21}(n)&=\Delta_{n,1}+\Delta_{n,2},\nonumber\\
P_{12}(n)&=\Delta_{n,1}q^n+\Delta_{n,2}.
\end{align}
\end{theorem}

\begin{proof} From Theorem  \ref{teo:quasi} (2),  we have that 
$L(\widehat{S\frak
S_q},\lambda)$ is quasifinite if and only if exist    $d\in (S
\frak S_q)_{-1/2}$ non-degenerate such that $\lambda([\widehat{(S
\frak S_q)}_{1/2},d])=0.$ But  by   la Proposition \ref{prop
p1}(c) this is equivalent to
\begin{align}
\quad 0=&\lambda
\big(f(T_q)b_{21}(T_q)I\big) \qquad\hbox{and}\,\nonumber\\
\quad0=&\lambda
\big(g(q^{-1}T_q)b_{12}(T_q)E_{22}+b_{12}(qT_q)g(T_q)E_{11}\big)-
\big(g(q^{-1}T_q)\,b_{12}(T_q)\big)_0\,c,
\end{align}
for all $f(w)$ y $g(w) \in \C[w,w^{-1}]$. In particular for
$f(w)=w^s$ and $g(w)=(qw)^r$, with $r,\,s\in\Z$, we have
\begin{align}
\,0=&\lambda
\left(T_q^sb_{21}(T_q)I\right), \nonumber\\
\,0=&\lambda
\left(T_q^rb_{12}(T_q)E_{22}+b_{12}(qT_q)(qT_q)^rE_{11}\right)-((T_q)^r\,b_{12}(T_q))_0\,c.
\end{align}
Writing $b_{12}(w)=\sum_i \beta_j^{12}w^j$ and $b_{21}(w)=\sum_i
\gamma_i^{21}w^i$ with $\beta_j^{12}$ and $\gamma_i^{21}
\in\mathbb{C}$,
\begin{align}\label{teoqf1}
\,0&=-\sum_i\gamma_i^{21}\lambda(T_q^{s+i}I)\nonumber\\
\,&=\sum_i\gamma_i^{21}(\Delta_{s+i,1}+\Delta_{s+i,2}).
\end{align}
and
\begin{align}\label{teoqf2}
\,0&=-\sum_j\beta_j^{12}\lambda(T_q^{r+j}E_{22}+(qT_q)^{r+j}E_{11})-\delta_{r,-j}\beta_j^{12}c\nonumber\\
\,&=\sum_j\beta_j^{12}(q^{r+j}\Delta_{r+j,1}+\Delta_{r+j,2})
-\delta_{r,-j}\beta_j^{12}c
\end{align}
Multiplying (\ref{teoqf1}) by $x^{-s}$ and adding over $ s\in\Z$,

\begin{align}
\,0&=\sum_{i,s}\gamma_i^{21}(\Delta_{s+i,1}x^{-s-i}+\Delta_{s+i,2}x^{-s-i})
x^i\nonumber\\
\,&=\sum_{i}\gamma_i^{21}\sum_s(\Delta_{s+i,1}x^{-s-i}+\Delta_{s+i,2}
x^{-s-i})x^i\nonumber\\
\,&=b_{21}(x)(\Delta_{\lambda,1}(x)+\Delta_{\lambda,2}(x)).
\end{align}
Similarly, multiplying (\ref{teoqf2}) by $x^{-r}$ and adding over
$r\in\Z$,
\begin{align}
\,0&=\sum_{j,r}\big[\beta_j^{12}(q^{r+j}\Delta_{r+j,1}x^{-r-j}+\Delta_{r+j,2}
x^{-r-j})x^j-\delta_{r,-j}\beta_j^{12}x^{-r}c\big]\nonumber\\
\,&=\sum_{j}\beta_j^{12}\big[\sum_r(q^{r+j}\Delta_{r+j,1}x^{-r-j}+\Delta_{r+j,2}
x^{-r-j})-c\big]x^j\nonumber\\
\,&=b_{12}(x)(\Delta_{\lambda,1}(q^{-1}x)+\Delta_{\lambda,2}(x)-c).
\end{align}
Thus we proved the first part. The equivalence between (i) and
(ii)follows from Theorem \ref{teo:quasipoly}.\end{proof}

\subsection {Interplay between $\widehat{S\frak S_q}$ and $\frak
{gl}_{\infty|\infty}{[m]}$ }

Given a non-negative integer $m$, consider the algebra of
truncated polynomials   $R=R_m=\C[t]/(t^{m+1})$, and let
$M_\infty[m] $ be the associative algebra consisting of matrices
$(a_{ij})_{i,j \in \Z}$ with $a_{ij} \in R_m$ such that $a_{ij}=0
$ for $|i-j|>>0$.
We  denote by  $\frak {gl}_{\infty}[m]$ the Lie algebra obtained from 
$M_\infty[m]$ by taking the usual commutator.

Define the associative superalgebra
$M_{\infty|\infty}[m]=M_{\infty}[m]\otimes M(1|1) $ with the
induced $\Z_2-$graded structure from $M(1\1).$

Denote $\frak {gl}_{\infty|\infty}[m]$ the Lie superalgebra
obtained from $M_{\infty|\infty}[m]$ by taking the usual super
commutator.
One may have two different ways of looking at  $\frak
{gl}_{\infty|\infty}[m].$ First we may regard  $\frak
{gl}_{\infty|\infty}[m]=\bigoplus_{i,j=1,2}\frak
{gl}_{\infty}[m]M_{ij},$ that is,
$$\frak
{gl}_{\infty|\infty}[m]=\left[\begin{array}{cc}
  \frak {gl}_\infty[m] & \frak {gl}_\infty[m] \\
  \frak {gl}_\infty[m] & \frak {gl}_\infty[m]
\end{array}\right].
$$

One also may identify
\begin{equation}\label{gl-infty}
\frak {gl}_{\infty|\infty}[m]=\{(a_{ij})_{i,j \in\Z/2} \,: \,
a_{ij} \in R_m \hbox{ and } a_{ij}=0 \hbox{ for } |i-j|>>0 \}.
\end{equation}
Under this identification, the   $\Z_2-$graded structure is given
by
$$ |E_{ij}|=\begin{cases}
\bar{0}, &\quad{\hbox{ if}} \quad {i-j \in\Z}\\
\,\bar{1} &\quad\hbox{ if} \, \quad{i-j\in\Z+1/2,}
\end{cases}
$$
where $E_{ij}$ denotes, as always, the infinite matrix with one in
the $ij$ entry and $0$ elsewhere.

The identification between  the two presentations of $\frak
{gl}_{\infty|\infty}[m]$ is given by ($i,\,j\in\mathbb{Z}$)
\begin{align}\label{identifglinfty}
\,E_{ij}M_{11}&=E_{i,j},\nonumber\\ \,
\,E_{ij}M_{22}&=E_{i-1/2,j-1/2},\nonumber\\
\,E_{ij}M_{12}&=E_{i,j-1/2},\nonumber\\
\,E_{ij}M_{21}&=E_{i-1/2,j}.
\end{align}

Under this identification, the Lie superalgebra $\frak
{gl}_{\infty|\infty}[m]$ is equipped with a natural
$\frac{1}{2}\Z-$gradation   $$\frak
{gl}_{\infty|\infty}[m]=\bigoplus_{r\in \frac{1}{2}
\Z}(\frak{gl}_{\infty|\infty}[m])_r$$  where
$(\frak{gl}_{\infty|\infty}[m])_r$ is the completion of the linear
span of $E_{ij}$ with $j-i=r$. This is the \textit{ principal
gradation } of $\frak {gl}_{\infty|\infty}[m]$.

Choose a  branch of $\log\, q$, and let $\tau=\frac{\log \,q}{2\pi
i}$. Then any $s\in\C$ is uniquely written as $s=q^a$, $a\in
\C/\tau^{-1}\Z$.

Take $s=q^a\in \C$ and let $R^{\infty|\infty}=R^\infty\bigoplus
R^\infty \theta=t^aR[t,t^{-1}]\bigoplus \theta t^aR[t,t^{-1}]$
with  $\theta$ an odd indeterminate. Consider the following basis
in   $R^{\infty|\infty}$,
$$
\{v_i=t^{-i+a}, \, v_{i-\frac {1}{2}}=t^{-i+a}\theta, \, i\in \Z
\}.$$ The Lie superalgebra  $\frak {gl}_{\infty|\infty}[m]$ acts
on $R^{\infty|\infty}$ by letting  $E_{ij}v_k=\delta_{jk}v_i$ with
$i,j,k \in \frac{1}{2}\Z.$ The Lie superalgebra  $S\frak S_q$ acts
on  $R^{\infty|\infty}$ as quantum pseudo-differential operators.
In this way we obtain a family of embeddings $\varphi_s^{[m]}$ of
$S\frak S_q$  into $\frak {gl}_{\infty|\infty}$
given by
\begin{align}
\,\varphi_s^{[m]}(t^k f_{11}(T_q)M_{11})&=\sum_{j\in 
\Z}f_{11}(sq^{-j+t})E_{j-k,j}, \nonumber\\
\,\varphi_s^{[m]}(t^k f_{21}(T_q)M_{21})&=\sum_{j\in 
\Z}f_{21}(sq^{-j+t})E_{j-k-\frac {1}{2},j}, \nonumber\\
\,\varphi_s^{[m]}(t^k f_{12}(T_q)M_{12})&=\sum_{j\in 
\Z}f_{12}(sq^{-j+t})E_{j-k,j-\frac{1}{2}}, \nonumber\\
\,\varphi_s^{[m]}(t^k f_{22}(T_q)M_{22})&=\sum_{j\in
\Z}f_{22}(sq^{-j+t})E_{j-k-\frac{1}{2},j-\frac{1}{2}}.
\end{align}

Note that the principal gradation on $\frak
{gl}_{\infty|\infty}[m]$ is compatible with that on $S\frak S_q$
under the map $\varphi_s^{[m]}$ and observe that the embedding $\varphi_s^{[m]}$ restricted to the
$\frak S_q M_{11}$ coincides with (6.2.1) in [2].

Denote by $\cal O$  the algebra of all holomorphic functions on
$\C^{\times}$ with topology of uniform convergence on compact
sets. We define a completion  $S\frak S_q^{a\cal O}$ of the
associative superalgebra of quantum pseudo-differential operators
by considering quantum pseudo-differential operators of infinite
order of the form $z^kf(T_q)M_{ij}$, where $f \in \cal O.$ The
embedding $\varphi_{s}^{[m]}$ extends naturally to $S\frak
S_q^{a\cal O}$.

Define
$$
\,I_s^{[m]}=\{ f \in {\cal O} \,:\, \,f^{(i)}(sq^n)=0 \,\hbox{ for
all } \, n\in\Z, \, i=0,\cdots ,m\}$$ and
$$\,J_s^{[m]}=\bigoplus_{i,j=1}^2\bigoplus_{k\in \Z} z^k
I_s^{[m]}M_{ij}.$$ Therefore, it follows by the Taylor formula for
$\varphi_s^{[m]}$ that
$$
\ker {\varphi}_s^{[m]}= J_s^{[m]}.
$$

Now,   fix  $\vec{s}=(s_1, \dotsb , s_n)\in\C^n $ such that if
we
write each $s_i=q^{a_i}$, we have
\begin{equation}\label{discreteseq}
a_i-a_j\notin \Z+\tau^{-1}\Z\hbox{          for  } i\neq j,
\end{equation}
and fix $\vec{m}=(m_1,\dotsb , m_n)\in \Z_+^n $.

Let $M_{\infty|\infty}[\vec{\hbox{m}}]= \oplus_{i=1}^n  M
_{\infty|\infty}[m_i]$. Consider the homomorphism
$$
\varphi_{\vec{s}}^{[\vec{m}]}=\bigoplus_{i=1}^n
\varphi_{s_i}^{[m_i]}:{S{\frak S}_q^{\cal O}}^{as}\longrightarrow  M
_{\infty|\infty}[\vec{m}].
$$
It is well known that for every discrete sequence of points in
$\C$ and a non negative integer $m$ there exists $f(w)\in\cal O$
having prescribed values of its first $m$ derivatives. Thus, due
to this fact and condition (\ref{discreteseq}) the following
Proposition follows.

\begin{proposition}\label{shortexactsequence}We have the exact sequence of 
$\frac{1}{2}\Bbb Z$-graded
associative superalgebras, provided that  $|q|\neq 1$:
$$
0\longrightarrow J_{\vec{s}}^{[\vec{m}]}\;\longrightarrow\;
{{S\frak
S}_q^{\cal O}}^{as}\;\overset{\,\,\,\varphi_{\vec{s}}^{[\vec{m}]}\,\,}
\longrightarrow \; M _{\infty|\infty}[\vec{m}] \longrightarrow  0
$$
where $J_{\vec{s}}^{[\vec{m}]}= \bigcap_{i=1}^n J_{s_i}^{[m_i]}$.

\end{proposition}

Consider the following \textit{super 2-cocycle} on  $\frak{
{gl}}_{\infty|\infty}[m]$ with values in $R_m:$
$$C(A,B)=Str([J,A],B), \quad A, B \in \frak {gl}
_{\infty|\infty}[m],$$  where $J=\sum_{r\leq 0}E_{r,r}$, and for a
matrix
$A=(a_{ij})_{i,j\in\frac{1}{2}\Z} \in \frak {gl}_{\infty|\infty}[m]$, 
$Str(A):=\sum_{r\in
\frac{1}{2}\Z}(-1)^{2r}a_{rr}$. Note that $C(A,B)$ is well defined for all 
$A$ and $B$ in $\frak {gl}_{\infty|\infty}[m]$. 
Denote by
$$\widehat{\frak {gl}}_{\infty |\infty}[m]=\{(a_{ij})_{i,j \in\Z/2} \,: \,
a_{ij} \in R_m \hbox{ y }  a_{ij}=0  \hbox{ para }
  |i-j|>>0 \}\bigoplus R_m,$$  the corresponding central extension. The
$\frac{1}{2}\mathbb{Z}$-gradation of this Lie superalgebra extends from 
$\frak {gl}_{\infty|\infty}[m]$ by letting  $gr(R_m )=0$. 

Therefore we have the following

\begin{lemma} The $\C$-linear map  
$\widehat{\varphi}_s^{[m]}:\widehat{S\frak S_q}\rightarrow \widehat{\frak 
{g}l}_{\infty|\infty}[m]$
defined by
\begin{align*}
\widehat{\varphi}_s^{[m]}(z^rT_q^kM_{ij})&=\varphi_s^{[m]}(z^rT_q^kM_{ij})
\quad\hbox{with} \quad r\neq 0,\nonumber\\
\, \widehat{\varphi}_s^{[m]}(C)&=1 \in R_m,\nonumber\\
\widehat{\varphi}_s^{[m]}(T_q^kM_{ii})&=\varphi_s^{[m]}(T_q^kM_{ii})-(-1)^i\frac{q^{ak}}{1-q^k}
\sum_{j=0}^m(k\log
q)\frac {t^j}{j!} \quad\hbox{if} \quad k\neq 0,\nonumber\\
\widehat{\varphi}_s^{[m]}(T_q^kM_{ij})&=\varphi_s^{[m]}(T_q^kM_{ij}) 
\quad\hbox{ for all}\quad k, \hbox{ and }i\neq j,\nonumber\\
\widehat{\varphi}_s^{[m]}(M_{ii})&=\varphi_s^{[m]}(M_{ii})
\end{align*}
is a homomorphism of Lie superalgebras. \end{lemma}

We shall need the following Proposition, whose proof is completely
similar to Proposition 4.3 in [2].

\begin{proposition}\label{prop432} Let $V$ be a quasifinite  
$\widehat{S{\frak S}}_q
$-module. Then the action of $\widehat{S{\frak S}}_q $ on $V$ naturally 
extends
to the action of $(\widehat{{S\frak S}}_q^{ \cal O})_k$ on $V$ for any 
$k\neq
0$.
\end{proposition}

We return now to the $\frac{1}{2}\Z$-graded complex Lie
superalgebra $\widehat{\frak {gl}}_{\infty|\infty}[m]$.

An element  $\lambda \in (\widehat{\frak {gl}}_{\infty |
\infty}[m])^*_0$ is characterized by its {\it labels}

\begin{equation}
\,\lambda_{k}^{(j)}=\lambda(t^jE_{kk}), \quad k\in
\frac{1}{2}\Z,\quad j=0\cdots,m
\end{equation}
and {\it central charges}
\begin{equation}
c_j=\lambda(t^j) \quad j=0,\cdots,m.
\end{equation}

As usual, we have the irreducible highest weight $\widehat{\frak
{gl}}_{\infty | \infty}[m]$-module \linebreak $L(\widehat{\frak
{gl}}_{\infty | \infty}[m],\lambda)$ associated to $\lambda$. We
will prove the following

\begin{theorem}\label{quasi-{gl}}  The $\widehat{\frak {gl}}_{\infty
| \infty}[m]-\hbox{module} \quad L(\lambda)$ is quasifinite if and
only if for each  $l=0\cdots,m$  all but finitely many $k\in
\frac{1}{2}\Z$,
\begin{equation}\label{eq-lambdas}
\lambda_k^{(l)}+\lambda_{k-\frac{1}{2}}^{(l)}+\delta_{k,\frac{1}{2}}c_l=0.
\end{equation}
\end{theorem}

\begin{remark} The case $m=0$ of this Theorem was proved in [8]. However 
using Theorem \ref{teo:quasi}, even this proof can be simplified.
\end{remark}
In order to apply Theorem \ref{teo:quasi} we need to show  that
the superalgebra $\widehat{\frak {gl}}_{\infty | \infty}[m]$
satisfies conditions ($SP_1$), ($SP_2$) and ($SP_3$) introduced in
Section 2.

The fact that $({\widehat{\frak {gl}}_{\infty | \infty}[m]})_0$ is
commutative is straightforward, thus we have ($SP_1$).

Let us check ($SP_2$).
Take $a=\sum_{j\in\frac{1}{2}\Z} a_j\,E_{j+k,j}\in(\widehat{\frak 
{gl}}_{\infty |\infty}[m])_{-k}$ with
$k\in \frac{1}{2}\Z, \, k>0$, such that $[a,b]=0$ for all $b\in
(\widehat{\frak {gl}}_{\infty | \infty}[m])_{\frac{1}{2}}$. In
particular for $b=E_{s-\frac{1}{2},s}$ for any $s\in
\frac{1}{2}\Z$. Thus we have that
$$0=a_{s-\frac{1}{2}}E_{s-\frac{1}{2}+k,s}-(-1)^{2k}a_{s-k}\,E_{s-\frac{1}{2},s-k}.$$
Since $k>0, \, E_{s+\frac{1}{2}-k,s}$ and $E_{s+\frac{1}{2},s+k}$
are   linearly independent, then $a_{s-\frac{1}{2}}=0$ and
$a_{s-k}=0$  for all $s\in\frac{1}{2}\Z$.
We conclude that $a=0$ proving $(SP_2)$.

\begin{remark}\label{degen}
Take $a=\sum_{j\in \frac{1}{2}\Z}a_j(t)\,E_{j+\frac{1}{2},\,j}\in
(\widehat{\frak {gl}}_{\infty | \infty}[m])_{-\frac{1}{2}}$ and
recall that by definition $$ \frak
p_{-\frac{1}{2}}^{a}=\sum[\cdots[[a,(\widehat{\frak {gl}}_{\infty
| \infty}[m])_0],(\widehat{\frak {gl}}_{\infty
|\infty}[m])_0]\cdots].$$

So letting $b=t^sE_{ii} \in (\widehat{\frak {gl}}_{\infty |\infty}[m])_0,$
we have in particular that\linebreak
$[a,b]=a_i(t)t^sE_{i+\frac{1}{2},i}-a_{i-\frac{1}{2}}(t)t^sE_{i,i-\frac{1}{2}}
\in \frak p_{-\frac{1}{2}}^{a}$. Then, for arbitrary $k$
\begin{align*}
\,[[a,b],E_{kk}]&=\delta_{ik}\big(a_k(t)t^sE_{k+\frac{1}{2},k}
+a_{k-\frac{1}{2}}(t)t^sE_{k,k-\frac{1}{2}}\big)
\nonumber\\
\,&-\delta_{i+\frac{1}{2},k}\,a_{k-\frac{1}{2}}(t)\,t^s\,E_{k,k-\frac{1}{2}}
-\delta_{i-\frac{1}{2},k}\,a_{k}(t)\,t^s\,E_{k+\frac{1}{2},k}\in
\frak p_{-\frac{1}{2}}^{a}.
\end{align*}
Choosing  $k=i+\frac{1}{2},$ we show that
$a_i(t)t^sE_{i+\frac{1}{2},i}\in \frak p _{-\frac{1}{2}}^{a}$ {for
all} $i\in\frac{1}{2}$ and  $s=0,\cdots,m.$

Let $I_{a_i(t)}$, $i\in\frac{1}{2}\mathbb{Z}$, the ideal of $R_m$
generated  by the corresponding $a_i(t)$. Thus we have shown that
\begin{equation}\label{degen-1}
\prod I_{a_i(t)}\,E_{i+\frac{1}{2},i}\subseteq \frak
p_{-\frac{1}{2}}^{a}.
\end{equation}
Computing the bracket between $a$ and an arbitrary element in
($\widehat{\frak{g}l}_{\infty|\infty}[m])_0$ it is easy to show
that equality holds in (\ref{degen-1}). Now, since $\frak
p_{-k-\frac{1}{2}}^a=[\frak p_{-k}^a,\frak p_{-\frac{1}{2}}^a]$,
inductively, it is straightforward to show that

\begin{equation}\label{degen-2}
\prod I_{a_i(t)}\,E_{i+k,i}\subseteq \frak p_{-k}^{a},
\end{equation}
for all $k\in\frac{1}{2}\mathbb{N}$.

\end{remark}
In order to check $(SP_3)$ first  we will describe the
non-degenerate elements of
$\widehat{\frak{g}l}_{\infty|\infty}[m]$ in the following

\begin{lemma}\label{degenerate-{gl}} An element   
$a=\sum_ja_j(t)E_{j+\frac{1}{2},j}\in (\widehat{\frak {gl}}_{\infty | 
\infty}[m])_{-\frac{1}{2}}$ is non-degenerate  if and only if  
$a_j(t)\in\C-\{0\}$ for all  but finitely many $j\in
\frac{1}{2}\Z$.
\end{lemma}

\begin{proof}
Suppose that $a$ is non-degenerate, that is, $\frak p_{-j}^{a}$
have  finite codimension in $(\widehat{\frak {gl}}_{\infty |
\infty}[m])_{-j}$ for all $j\in\frac{1}{2}\mathbb{N}$. In
particular $\frak p_{-\frac{1}{2}}^{a}$ have finite codimension in
$(\widehat{\frak {gl}}_{\infty | \infty}[m])_{-\frac{1}{2}}$.
Thus, since equality holds in (\ref{degen-1}), $a_i(t)\in
\mathbb{C}-\{0\}$ for all but  finitely many $i\in \frac{1}{2}\Z$.

The converse statement follows immediately  from (\ref{degen-1})
and (\ref{degen-2}).\end{proof}

Let us check $(SP_3)$. Let $\frak p$ non-degenerate. In particular
$\frak p_{-\frac{1}{2}}$ have finite codimension on
$({\widehat{\frak {gl}}_{\infty |\infty}[m])_{-\frac{1}{2}}}$.
Thus $\frak p_{-\frac{1}{2}}=\prod_i I_iE_{i+\frac{1}{2},i}$ with
$I_i\subseteq R_m$ a subspace such that $I_i=R_m$ for all  but  finitely
many $i$. Let $K$ be such finite set. Thus, by Lemma
\ref{degenerate-{gl}},  $a=\sum_{j\in K^c}E_{j-\frac{1}{2},j} \in
\frak p_{-\frac{1}{2}}$ is non-degenerate and by definition $\frak
p^a\subseteq \frak p$.

Now we can prove Theorem \ref{quasi-{gl}}.

\begin{proof} Let $\lambda \in (\widehat{{\frak{gl}}}_{\infty| \infty}[m])_0 
^*$.
By Theorem \ref{teo:quasi} and Lemma \ref{degenerate-{gl}},
$L(\lambda)$ is quasifinite if and only if  there exists $a=\sum_i
a_i(t)E_{i+\frac{1}{2},i}\in(\widehat{{\frak{gl}}}_{\infty |
\infty}[m])_{-\frac{1}{2}}$ with $ a_i\in
\mathbb{C}-\{0\} $ for all but finitely many $i$ and
$\lambda([a,b])=0$ for all $b\in (\widehat{{\frak{gl}}}_{\infty|
\infty}[m])_{\frac{1}{2}}.$

Suppose that $L(\lambda)$ is quasifinite, thus an element $a$ as
above exists. Let $I=\{i\in \frac{1}{2}\Z: \,a_i  \not\in
\mathbb{C}-\{0\}\}.$ Note that $|I|<\infty.$ Consider $k\in I^c,
|k|>>0$ such that if $k \in I^c,$ then $k-\frac{1}{2} \in I^c.$

Take  $b=t^lE_{k-\frac{1}{2},k} \in ({\widehat{{\frak
{gl}}}_{\infty | \infty}[m])_{\frac{1}{2}}}$ with $l=0,\cdots,m$.

Since
$[a,b]=a_{k-\frac{1}{2}}t^l\left(E_{kk}+E_{k-\frac{1}{2},k-\frac{1}{2}}+\delta_{k,\frac{1}{2}}\right)$,
then $\lambda([a,b])=0$ implies
\begin{align*}
0&=\lambda\left(t^lE_{k,k}+t^lE_{k-\frac{1}{2},k-\frac{1}{2}}+t^l\delta_{k,\frac{1}{2}}\right)\\
&=\lambda_k^{(l)}+\lambda_{k-\frac{1}{2}}^{(l)}
+\delta_{k,\frac{1}{2}}c_l\quad\hbox{ for all} \quad k\in I^c, \,
\hbox{ with} \quad|k|>>0.
\end{align*}
Hence,  quasifiniteness of  the $\widehat{\frak {gl}}_{\infty |
\infty}[m]$-module $L(\lambda)$   implies (\ref{eq-lambdas}).

Conversely,  assume that (\ref{eq-lambdas}) holds. Denote by $I$
the finite set where this condition is not satisfied.
Let $0<<N$   such that if
$i \in I^c$ and $|i|>N,$ then $i\pm \frac{1}{2}\in I^c.$ Set
$a=\sum_{|i|>N}E_{i,i-\frac{1}{2}} \in {\widehat{\frak
{gl}}_{\infty |\infty}[m]_{-\frac{1}{2}}},$ . By Lemma
\ref{degenerate-{gl}}  $a$ is non-degenerate.

Consider an arbitrary element $b=\sum_jb_j(t)E_{j-\frac{1}{2},j}
\in {\widehat{\frak {gl}}_{\infty | \infty}[m]_{\frac{1}{2}}}$.
Write each $b_j(t)=\sum_{l=0}^m B_l^j\,t^l$. Then
\begin{align*}
\quad \lambda([a,b])&=\lambda\Big(\sum_{j,
|i|>N}b_j(t)[E_{i,i-\frac{1}{2}},E_{j-\frac{1}{2},j}]+\sum_{j,|i|>N}b_j(t)C(E_{i,i-\frac{1}{2}},E_{j-\frac{1}{2},j})\Big)\nonumber\\
\,&=\lambda\Big(\sum_{
|i|>N}b_{i}(t)(E_{i-\frac{1}{2},i-\frac{1}{2}}+E_{i,i})+\sum_{|i|>N}b_{i}(t)C(E_{i,i-\frac{1}{2}},E_{i-\frac{1}{2},i})\Big)\nonumber\\
&\,=\sum_{l=0}^m\sum_{
|i|>N}B_l^{i}\Big(\lambda(t^l\,E_{i-\frac{1}{2},i-\frac{1}{2}})+\lambda(t^lE_{i,i}))+\delta_{i,\frac{1}{2}}c_l\Big)\nonumber\\
\,&=0
\end{align*}
finishing the proof.\end{proof}

Given $\vec{m}=(m_1,\cdots,m_N)\in \Z_+^N,$ we define
$\widehat{\frak
{gl}}_{\infty|\infty}[\vec{m}]=\bigoplus_{i=1}^N\widehat{\frak
{gl}}_{\infty|\infty}[m_i] $. By Proposition
\ref{shortexactsequence} we have a  surjective Lie superalgebra
homomorphism
$$\widehat{\varphi}_{\vec{s}}^{[\vec{m}]}=\bigoplus_{i=1}^N\widehat{\varphi}_{s_i}^{[m_i]}:\widehat{S\frak S_q}^{\cal O}\rightarrow \widehat{\frak
{gl}}_{\infty|\infty}[\vec{m}],$$

Choose a quasifinite   $\lambda_i \in (\widehat{\frak
{gl}}_{\infty|\infty}[m_i])_0^{\star}$ and let $L(\widehat{\frak
{gl}}_{\infty|\infty}[m_i],\lambda_i)$ be the corresponding
irreducible  $\widehat{\frak {gl}}_{\infty|\infty}[m_i]$-module.
Then
$$L(\widehat{\frak
{gl}}_{\infty|\infty}[\vec{m}],\vec{\lambda})=\bigotimes_{i=1}^nL(\widehat{\frak
{gl}}_{\infty|\infty}[m_i],\lambda_i)$$ where
$\vec{\lambda}=(\lambda_1,\cdots,\lambda_N)$ is an irreducible
$\widehat{\frak {gl}}_{\infty|\infty}[\vec{m}]$-module. Using the
homomorphism $\widehat{\varphi}_{\vec{s}}^{\vec{m}},$ we regard
$L(\widehat{\frak {gl}}_{\infty|\infty}[\vec{m}],\vec{\lambda})$
as a $\widehat{S\frak S_q}$-module,  which we denote by
$L_{\vec{s}}^{[\vec{m}]}(\vec{\lambda}).$

\begin{theorem}\label{teo:MI} Consider the embedding
$\widehat{\varphi}_{\vec{s}}^{[\vec{m}]}:\widehat{S\frak
S_q}\rightarrow \widehat{\frak {gl}}_{\infty|\infty}[\vec{m}]$
where $\vec{s}=(s_1,\cdots,s_N),$ with $s_i=q^{a_i}\in \mathbb{C}$
such that
$a_i-a_j \not\in\Z+\tau^{-1}\Z$ if $i\neq j$, and let $V$ be a quasifinite
$\widehat{\frak {gl}}_{\infty|\infty}[\vec{m}]$-module. Then any
$\widehat{S\frak S_q}$-submodule of $V$ is a $\widehat{\frak
{gl}}_{\infty|\infty}[\vec{m}]$-submodule as well. In particular,
the $\widehat{S\frak S_q}$-modules
$L_{\vec{s}}^{[\vec{m}]}(\vec{\lambda})$ are irreducible.

\begin{proof}Let  $U$ be a ($\frac{1}{2}\Z$-graded) $\widehat{S\frak
S_q}$-submodule de $V.$ $U$ is a quasifinite $\widehat{S\frak
S_q}-$module as well, hence by Proposition \ref{prop432}, it can
be extended to $(\widehat{S\frak S_q}^{\cal O})_k$ for any $k\neq
0.$ But the map $\widehat{\varphi}_{\vec{s}}^{[\vec{m}]}$ is
surjective for any $k\neq0.$ Thus  $U$ is invariant with respect
to all members of the principal gradations $(\frak {gl}
_{\infty|\infty}[\vec{m}])_k$ with $k\neq0.$ Since $\widehat{\frak
{gl}}_{\infty|\infty}[\vec{m}]$ coincides with its derived
algebra, this proves the theorem.
\end{proof}
\end{theorem}

By Proposition  \ref{prop432} and Theorem \ref{teo:MI}, the
$\widehat{S\frak S_q}$-modules
$L_{\vec{s}}^{[\vec{m}]}(\vec{\lambda})$ are irreducible
quasifinite highest weight modules. Let us calculate the labels
$\Delta_{m,s,\lambda,k,i}$, with $i=1,2$, of the highest weight
and the central charge $c$ of the $\widehat{S\frak S_q}$-modules
$L_{{s}}^{[{m}]}({\lambda})$. We have ($k\neq 0$)
\begin{align}
\Delta_{m,s,\lambda,k,1}&=\lambda 
\big(\widehat{\varphi}_{s}^{[m]}(T_q^kM_{11})\big)\nonumber\\
\,&=\sum_{l=0}^m\frac{(k\log
q)^l}{l!}\Big[\sum_j(sq^{-j})^k\lambda_j^{(l)}+\frac{q^{ak}}{1-q^k}c_l\Big],
\end{align}

\begin{align}
\Delta_{m,s,\lambda,k,2}&=\lambda 
\big(\widehat{\varphi}_{s}^{[m]}(T_q^kEM_{22})\big)\nonumber\\
\,&=\sum_{l=0}^m\frac{(k\log
q)^l}{l!}\Big[\sum_j(sq^{-j})^k\lambda^{(l)}_{j-\frac{1}{2}}-\frac{q^{ak}}{1-q^k}c_l\Big]
\end{align} 
and for $k=0$
$$\Delta_{m,s,\lambda,0,i}=\lambda(\widehat{\varphi}_{s}^{[m]}(E_{ii}))=\begin{cases}
\sum_j\lambda(E_{jj})=\sum_j\lambda^{(0)}_j &\quad\hbox{ if} \quad
i=1\nonumber\\
\sum_j\lambda(E_{j-\frac{1}{2}j-\frac{1}{2}})=\sum_j\lambda^{(0)}_{j-\frac{1}{2}}
&\quad\hbox{ if} \quad i=2.\nonumber\\
\end{cases}$$

The following Theorem shows that any irreducible quasifinite
highest weight module $L(\widehat{S\frak S_q},\lambda)$ can be
obtained in a unique way. The proof of this result follows by the
same argument used in Theorem 4.8 in [2] using the formulas
above.

\begin{theorem} Let $L=L(\widehat{S\frak S_q},\lambda)$ be an
irreducible quasifinite highest weight module with central charge c and
$$\Delta_{n,1}+\Delta_{n,2}=P_{21}(n)\quad\hbox{ and }\quad 
\Delta_{n,2}=\frac{P_{12}(n)-P_{21}(n)}{q^n-1}$$
for $n\neq 0$, where $P_{12}(x)$ y $P_{21}(x)$ are quasipolynomial
such that $P_{12}(0)-P_{21}(0)=c.$ We write $P_{ij}(x)=\sum_{a \in
\mathbb{C}} P_{ij,a}(x\log q)\,q^{ax} \hbox{ where }
P_{ij,a}(x\log q)\,$   are polynomials . We decompose the set
$\{a\in \mathbb{C}\,:\, P_{ij,a}(x\log q)\neq 0 \} $ in a disjoint
union of congruence classes  $\hbox{mod }\, \Z+\tau^{-1}\Z.$ Let
$S=\{a,a-k_1,a-k_2,\cdots \}$ be such a congruence class,  let
$m=\max_{a\in S} \deg P_{ij,a}(x\log q),$ and let \linebreak
$$
h^{(l)}_{k_r-\frac{1}{2}}=\left(\frac{d}{dx}\right)^l\,P_{21,a-k_r}(0)
\qquad \hbox{ and }\qquad
h^{(l)}_{k_r}=\left(\frac{d}{dx}\right)^l\,P_{12,a-k_r}(0) .$$ We
associate to $S$ the $\,\widehat{\frak
{gl}}_{\infty|\infty}[m]$-module $L^{[m]}(\lambda_s)$ with the
central charges
\begin{equation}
c_l=\sum_{k_r}(h^{(l)}_{k_r-\frac{1}{2}}-h^{(l)}_{k_r}),
\end{equation}
and labels
\begin{equation}
\lambda_{i}^{(l)}=\sum_{k_r>i}
\left(\widetilde{h}^{(l)}_{k_r}-h^{(l)}_{k_r-\frac{1}{2}}\right)\qquad\hbox{
and }\qquad\lambda_{i-\frac{1}{2}}^{(l)}=\sum_{k_r\geq i}
\left(h^{(l)}_{k_r-\frac{1}{2}}-
\widetilde{h}^{(l)}_{k_r+1}\right)
\end{equation}
for $i\in\Z$ and $\widetilde{h}^{(l)}_{k}=h^{(l)}_{k}+\delta_{k,0}
c_l$. Then the  $\widehat{S\frak S_q}$-module $L$ is isomorphic to
the tensor product of all the modules  $L_s^{[m]}(\lambda_S)$.

\end{theorem}

\end{document}